\documentclass[sigconf, 10pt, papersize={8.5in,11in}, nonacm]{acmart}

\usepackage{booktabs} 

\setcopyright{rightsretained}

\acmConference[ICCPS'19]{ACM/IEEE}{April 2019}{Canada}

\begin{document}
	\title{Two-Way Coding in Control Systems Under Injection Attacks: From Attack Detection to Attack Correction}

\author{Song Fang}
\affiliation{%
  \institution{School of Electrical Engineering and Computer Science, KTH Royal Institute of Technology}
  \city{Stockholm}
  \country{Sweden}
}
\email{sonf@kth.se}

\author{Karl Henrik Johansson}
\affiliation{%
	\institution{School of Electrical Engineering and Computer Science, KTH Royal Institute of Technology}
	\city{Stockholm}
	\country{Sweden}
}
\email{kallej@kth.se}

\author{Mikael Skoglund}
\affiliation{%
	\institution{School of Electrical Engineering and Computer Science, KTH Royal Institute of Technology}
	\city{Stockholm}
	\country{Sweden}
}
\email{skoglund@kth.se}

\author{Henrik Sandberg}
\affiliation{%
	\institution{School of Electrical Engineering and Computer Science, KTH Royal Institute of Technology}
	\city{Stockholm}
	\country{Sweden}
}
\email{hsan@kth.se}

\author{Hideaki Ishii}
\affiliation{%
	\institution{Department of Computer Science, Tokyo Institute of Technology}
	\city{Yokohama}
	\country{Japan}
}
\email{ishii@c.titech.ac.jp}

%
%
%
%
%


\begin{abstract}
In this paper, we introduce the method of two-way coding, a concept originating in communication theory characterizing coding schemes for two-way channels, into (networked) feedback control systems under injection attacks. We first show that the presence of two-way coding can distort the perspective of the attacker on the control system. In general, the distorted viewpoint on the attacker side as a consequence of two-way coding will facilitate detecting the attacks, or restricting what the attacker can do, or even correcting the attack effect. In the particular case of zero-dynamics attacks, if the attacks are to be designed according to the original plant, then they will be easily detected; while if the attacks are designed with respect to the equivalent plant as viewed by the attacker, then under the additional assumption that the plant is stabilizable by static output feedback, the attack effect may be corrected in steady state.  
\end{abstract}

%
%
%

\keywords{Cyber-physical system, networked control system, cyber-physical security, two-way channel, two-way coding, zero-dynamics attack}

\maketitle

%


\section{Introduction}

The concept of two-way communication channels dates back to Shannon \cite{shannon1961two}. As its name indicates, in two-way channels, signals are transmitted simultaneously in both directions between the two terminals of communication. Accordingly, coding for two-way channels should make use of the information contained in the data transmitted in both directions; in other words, the coding schemes are also two-way, and thus are referred to as two-way coding \cite{van1977survey, meeuwissen1998information, chaaban2015multi}.

Inherently, the communication channels in networked feedback control systems are two-way channels, with the controller side and the plant side being viewed as the two terminals of communication, respectively. Nevertheless, approaches based on two-way coding for the two-way channels in networked feedback systems are rarely seen in the literature. One exception is the so-called scattering transformation utilized in the tele-operation of robotics \cite{anderson1989bilateral, niemeyer1991stable, hokayem2006bilateral, nuno2011passivity, hirche2012human, hatanaka2015passivity}; in a broad sense, scattering transformation can be viewed as a special class of two-way coding to resolve the issue of two-way time delays, the most essential characterization and the main issue of the two-way channels modeled on the input-output level in the problem of tele-operation. Other related applications of the scattering transformation include \cite{kimura1996chain, kailath2000linear, gu2011two}.

Particularly in the cyber-physical security problems (see, e.g., \cite{poovendran2012special, johansson2014guest, sandberg2015cyberphysical, teixeira2015secure, zhu2015game, amin2015game, smith2015covert, mo2015physical, pasqualetti2015control, cheng2017guest, giraldo2018survey} and the references therein) of networked control systems, to the best of our knowledge, only one-way coding has been employed. The authors of \cite{xu2015secure} introduced (one-way) encryption matrices into control systems to achieve confidentiality and integrity. In \cite{miao2017coding}, the authors considered a method of coding (using one-way coding matrices) the sensor outputs in order to detect stealthy false data injection attacks in cyber-physical systems. 
Modulation matrices, which are one-way, were inserted into cyber-physical systems in \cite{hoehn2016detection} to detect covert attacks and zero-dynamics attacks. 
Dynamic one-way coding was applied to detect and isolate routing attacks \cite{ferrari2017detection1} and replay attacks \cite{ferrari2017detection2}. 
For remote state estimation in the presence of eavesdroppers, the so-called state-secrecy codes were introduced \cite{tsiamis2017state}, which are also inherently one-way coding schemes.
On the other hand, as will be discussed in Section~4 of this paper, one-way coding has its inherent limitations; for instance, one-way coding in general cannot eliminate the unstable poles nor nonminimum-phase zeros of the plant nor the controller, which are most critical issues in the defense against, e.g., zero-dynamics attacks \cite{teixeira2015secure}. 

In this paper, we investigate how two-way coding can play an important role in protecting the security of feedback control systems under injection attacks. We first introduce a series of special classes of two-way coding, including the two-way stretching, shearing, and rotation matrices, as well as the scattering transformation. We then examine what changes the presence of two-way coding will bring to the feedback control system. On one hand, it is seen that on the controller and reference side, the plant behaves exactly as if two-way coding does not exist; as such, the controller may be designed regardless of two-way coding. On the other, two-way coding will distort the attacker's perspective of the signals and systems, i.e., the components of the feedback loop, giving him/her a ``transformed" view of the control system, and making the behaviors of the plant, controller, and reference all seemingly different from the those of the original system without two-way coding. 


More specifically, we examine how the presence of two-way coding can play a critical role in the defense against injection attacks. In general, the distorted perspective on the attacker side as a result of two-way coding will enable detecting the attacks or restricting what the attacker can do or even correcting the attack effect, depending on the attacker's knowledge of the system. As a matter of fact, two-way coding can make the zeros and/or poles of the equivalent plant as viewed by the attacker all different from those of the original plant, and under some additional assumptions (i.e., the plant is stabilizable by static output feedback), the equivalent plant may even be made stable and/or minimum-phase. In the particular case of zero-dynamics attacks, it is then implicated that the attacks will be detected if designed according to the original plant, while the attack effect may be corrected in steady state if the attacks are to be designed with respect to the equivalent plant.

The remainder of the paper is organized as follows. Section~2 is devoted to two-way coding. In Section~3, we introduce two-way coding into linear time-invariant (LTI) feedback control systems under injection attacks, and show how its presence can distort the perspective of the attacker. Section~4 analyzes the role two-way coding can play in the defense against injection attacks, in particular, zero-dynamics attacks. 
Concluding remarks are given in Section~5.

\section{Two-Way Coding}

Consider the single-input single-output (SISO) system depicted in Fig.~\ref{figure1}. Herein, $K$ denotes the controller while $P$ denotes the plant. The reference signal is $r \left( t \right) \in \mathbb{R}$ and the plant output is $\overline{y} \left( t \right) \in \mathbb{R}$. In addition, let $u \left( t \right)$, $\overline{u} \left( t \right)$, $y \left( t \right)$, $q \left( t \right)$, $\overline{q} \left( t \right)$, $v \left( t \right)$, $\overline{v} \left( t \right) \in \mathbb{R}$. 

\begin{figure}
	\vspace*{-0mm}
	\begin{center}
		\includegraphics [width=0.5\textwidth]{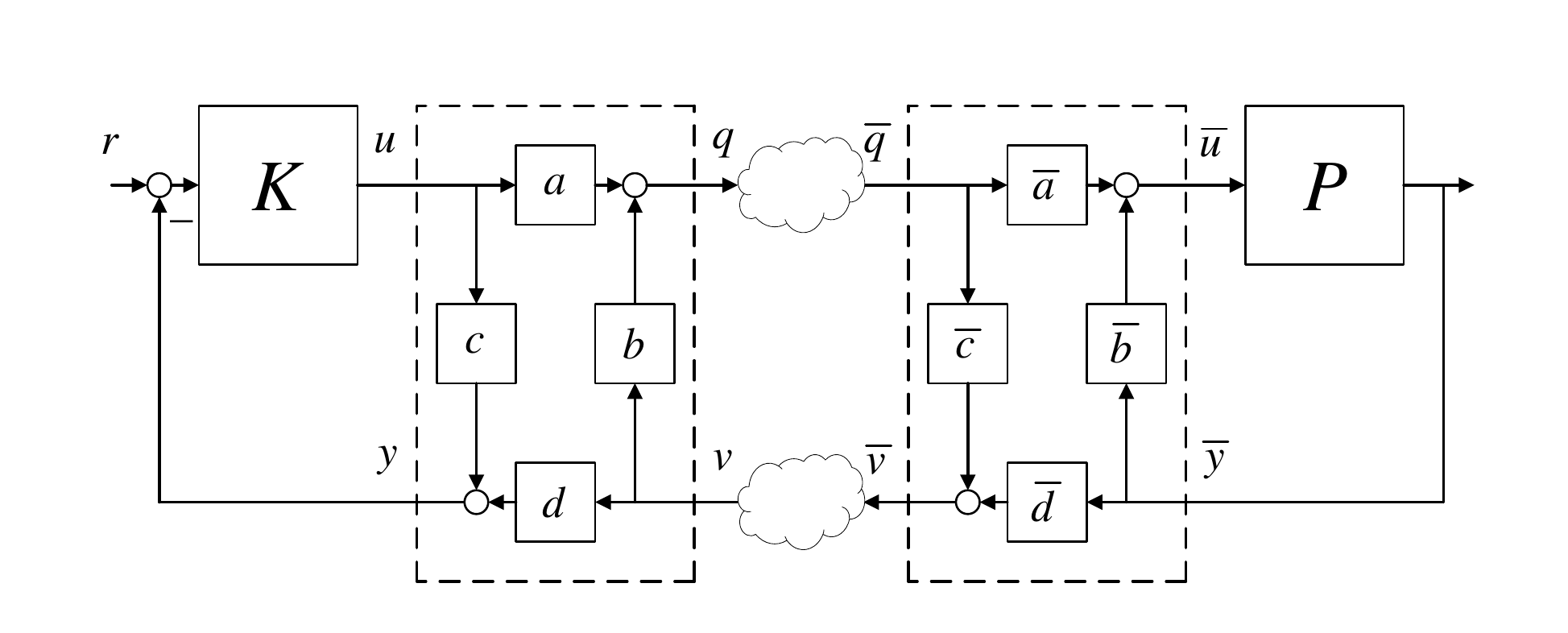}
		\vspace*{-9mm}
		\caption{A networked feedback system with two-way coding.}
		\label{figure1}
	\end{center}
	\vspace*{-3mm}
\end{figure}

\begin{definition}
	The (static) \textbf{two-way coding} is defined as
	\begin{flalign}
	\left[
	\begin{array}{c}
	q \left( t \right)\\
	y \left( t \right)\\
	\end{array}
	\right]
	&=M
	\left[
	\begin{array}{c}
	u \left( t \right)\\
	v \left( t \right)\\
	\end{array}
	\right],
	\end{flalign}
	where 
	\begin{flalign}
	M
	=\left[
	\begin{array}{cc}
	a & b \\
	c & d \\
	\end{array}\right].
	\end{flalign}
	Herein, $a, b, c, d \in \mathbb{R}$ are chosen such that 
	\begin{flalign} \label{condition1}
	ad \neq 0,~ ad - bc \neq 0. 
	\end{flalign}
	Strictly speaking, it should be further assumed that $\left| ad - bc \right| < \infty$.
\end{definition}

\vspace*{0mm}

Herein, two-way coding (operating in a feedback loop) represents a two-way transformation that takes in the signal in the forward path and the signal in the feedback path, and outputs a new signal to the forward path and a second new signal that passes on in the feedback path. In comparison, Fig.~\ref{figureoneway} depicts a system with one-way coding schemes, which are one-way transformations that either take in the signal in the forward path and output a new signal that passes on in the forward path, or input the signal in the feedback path and output a signal that continues in the feedback path; herein, $\alpha, \beta \in \mathbb{R}$ and  $0< \left| \alpha \right|, \left| \beta \right| < \infty$.

\begin{figure}
	\vspace*{-0mm}
	\begin{center}
		\includegraphics [width=0.5\textwidth]{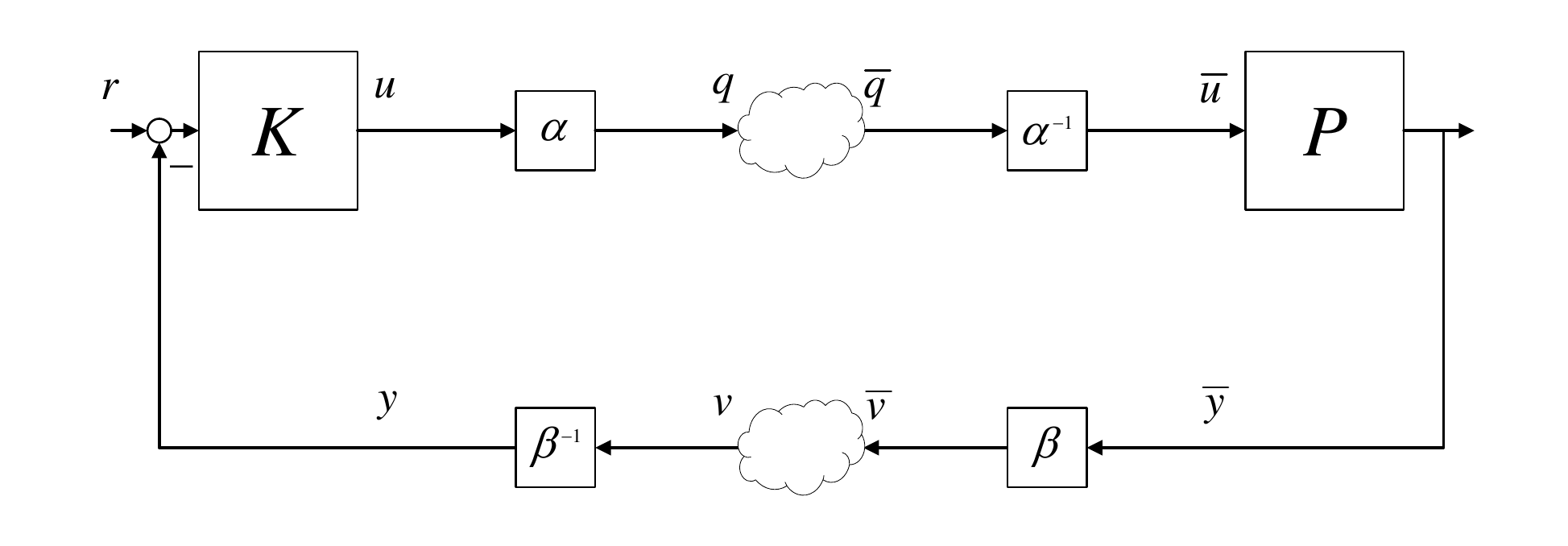}
		\vspace*{-9mm}
		\caption{A networked feedback system with one-way coding.}
		\label{figureoneway}
	\end{center}
	\vspace*{-3mm}
\end{figure}

For simplicity, we denote the inverse of two-way coding $M$ as 
\begin{flalign} \label{inverse}
\left[
\begin{array}{cc}
\overline{a} & \overline{b} \\
\overline{c} & \overline{d} \\
\end{array}\right]
= M^{-1} = \left[
\begin{array}{cc}
\frac{d}{ad-bc} & -\frac{b}{ad-bc} \\
-\frac{c}{ad-bc} & \frac{a}{ad-bc} \\
\end{array}\right],
\end{flalign}
where $\overline{a}, \overline{b}, \overline{c}, \overline{d} \in \mathbb{R}$. As illustrated on the plant side in Fig.~\ref{figure1}, the inverse of two-way coding $M$ denotes another two-way coding.

At this point, we do not impose any assumptions on the controller $K$ and plant $P$ except that the closed-loop system is stable; we now prove the following result for this generic setting.


\begin{proposition} \label{cancel}
	If $\overline{q} \left( t \right) = q \left( t \right)$ and $v \left( t \right) = \overline{v} \left( t \right)$, then $\overline{u} \left( t \right) = u \left( t \right)$ and $y \left( t \right) = \overline{y} \left( t \right)$.
\end{proposition}

\begin{proof}
	Since
	\begin{flalign}
	\left[
	\begin{array}{c}
	q \left( t \right)\\
	y \left( t \right)\\
	\end{array}
	\right]
	=\left[
	\begin{array}{cc}
	a & b \\
	c & d \\
	\end{array}\right]
	\left[
	\begin{array}{c}
	u \left( t \right)\\
	v \left( t \right)\\
	\end{array}
	\right], \nonumber
	\end{flalign}
	we have
	\begin{flalign}
	u \left( t \right) 
	= a^{-1} q \left( t \right) - a^{-1} b v \left( t \right), \nonumber
	\end{flalign}
	and
	\begin{flalign}
	y \left( t \right) 
	= c u \left( t \right) + d v \left( t \right)
	= c a^{-1} q \left( t \right) + \left( d - c a^{-1} b \right) v \left( t \right). \nonumber
	\end{flalign}
	Similarly, since
	\begin{flalign}
	\left[
	\begin{array}{c}
	\overline{u} \left( t \right)\\
	\overline{v} \left( t \right)\\
	\end{array}
	\right]
	=\left[
	\begin{array}{cc}
	\overline{a} & \overline{b} \\
	\overline{c} & \overline{d} \\
	\end{array}\right]
	\left[
	\begin{array}{c}
	\overline{q} \left( t \right)\\
	\overline{y} \left( t \right)\\
	\end{array}
	\right], \nonumber
	\end{flalign}
	and noting \eqref{inverse}, we have
	\begin{flalign}
	\overline{y} \left( t \right) 
	= - \overline{d}^{-1} \overline{c} \overline{q} \left( t \right) + \overline{d}^{-1} \overline{v} \left( t \right) 
	= c a^{-1} \overline{q} \left( t \right) + \left( d - c a^{-1} b \right) \overline{v} \left( t \right), \nonumber
	\end{flalign}
	and
	\begin{flalign}
	\overline{u} \left( t \right) 
	&= \overline{a} \overline{q} \left( t \right) + \overline{b} \overline{y} \left( t \right)
	= \left( \overline{a} - \overline{b} \overline{d}^{-1} \overline{c} \right) \overline{q} \left( t \right) + \overline{b} \overline{d}^{-1} \overline{v} \left( t \right) \nonumber \\
	&= a^{-1} \overline{q} \left( t \right) - a^{-1} b \overline{v} \left( t \right). \nonumber
	\end{flalign}
	Clearly, when $\overline{q} \left( t \right) = q \left( t \right)$ and $v \left( t \right) = \overline{v} \left( t \right)$, it follows that $\overline{u} \left( t \right) = u \left( t \right)$ and $y \left( t \right) = \overline{y} \left( t \right)$.
\end{proof}

In other words, if $\overline{q} \left( t \right) = q \left( t \right)$ and $v \left( t \right) = \overline{v} \left( t \right)$, the system in Fig.~\ref{figure1}, now equivalent to that of Fig.~\ref{figure2}, reduces to the system depicted in Fig.~\ref{figure3} as the ``original" feedback system without two-way coding. As such, properties, including stability and performance, of the system in Fig.~\ref{figure1} when $\overline{q} \left( t \right) = q \left( t \right)$ and $v \left( t \right) = \overline{v} \left( t \right)$ are equivalent to those of the original system in Fig.~\ref{figure3}.


\begin{figure}
	\vspace*{-0mm}
	\begin{center}
		\includegraphics [width=0.5\textwidth]{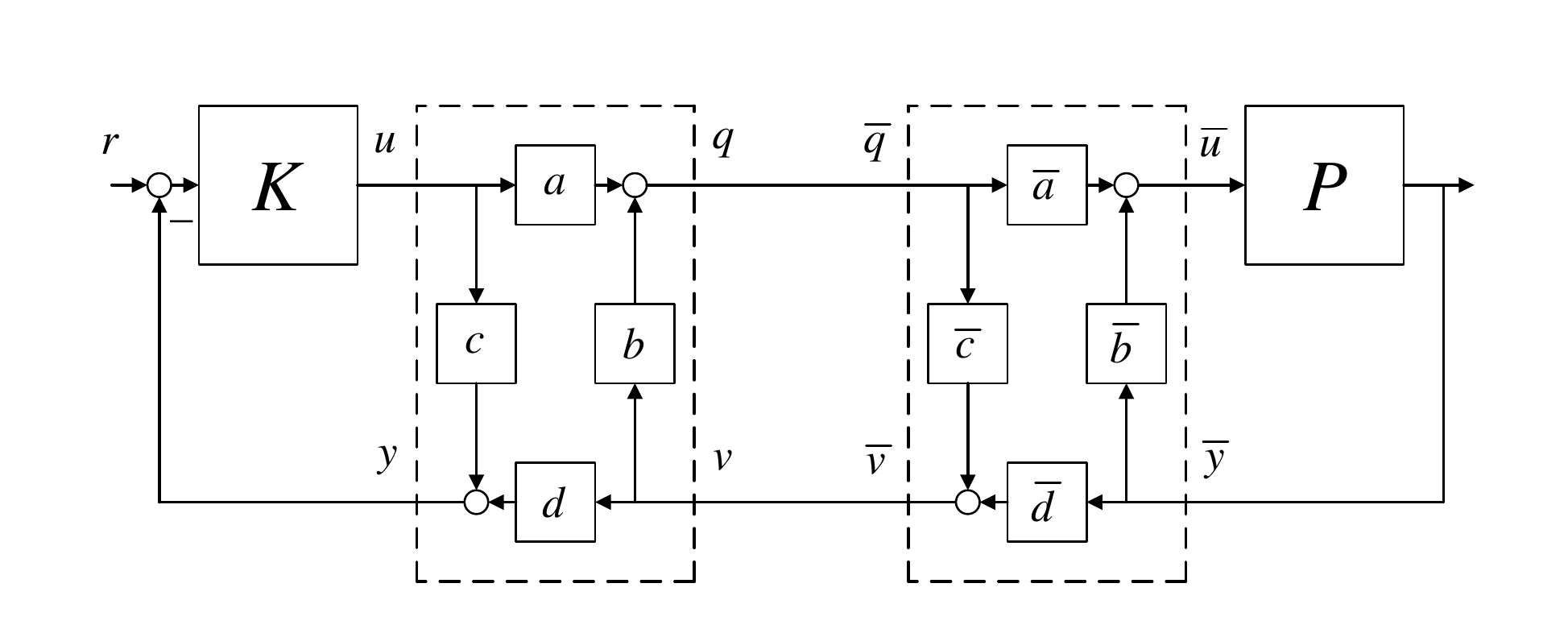}
		\vspace*{-9mm}
		\caption{A feedback system with two-way coding.}
		\label{figure2}
	\end{center}
	\vspace*{-3mm}
\end{figure}

\begin{figure}
	\vspace*{0mm}
	\begin{center}
		\includegraphics [width=0.25\textwidth]{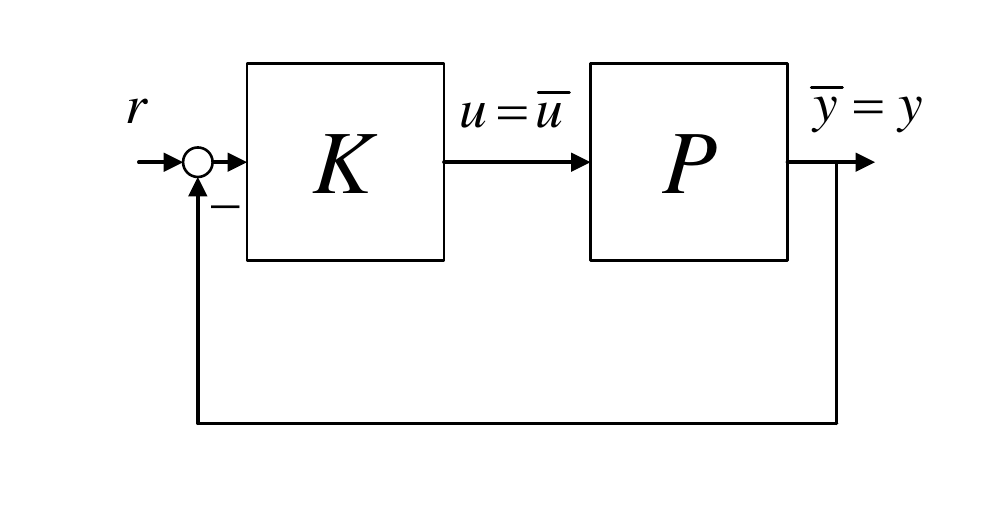}
		\vspace*{-3mm}
		\caption{The original feedback system without two-way coding.}
		\label{figure3}
	\end{center}
	\vspace*{-3mm}
\end{figure}

\subsection{Special Cases of Two-Way Coding}

We now consider some special cases of two-way coding matrices. 
In what follows, we will introduce the (two-way) stretching matrix, shearing matrix, rotation matrix, and so on that are adapted from 2D computer graphics \cite{hughes2014computer}, as well as the scattering transformation from tele-operation \cite{anderson1989bilateral, niemeyer1991stable, hokayem2006bilateral, nuno2011passivity, hirche2012human, hatanaka2015passivity}.

\subsubsection{Two-Way Stretching Matrix} \label{specialtwoway}
Below we list three cases of the two-way stretching matrices. 

Case 1:
\begin{flalign}
M
=\left[
\begin{array}{cc}
a & 0 \\
0 & 1 \\
\end{array}\right],~a \neq 0. \nonumber
\end{flalign}

Case 2:
\begin{flalign}
M
=\left[
\begin{array}{cc}
1 & 0 \\
0 & d \\
\end{array}\right],~d \neq 0. \nonumber
\end{flalign}

\begin{figure}
	\vspace*{-0mm}
	\begin{center}
		\includegraphics [width=0.5\textwidth]{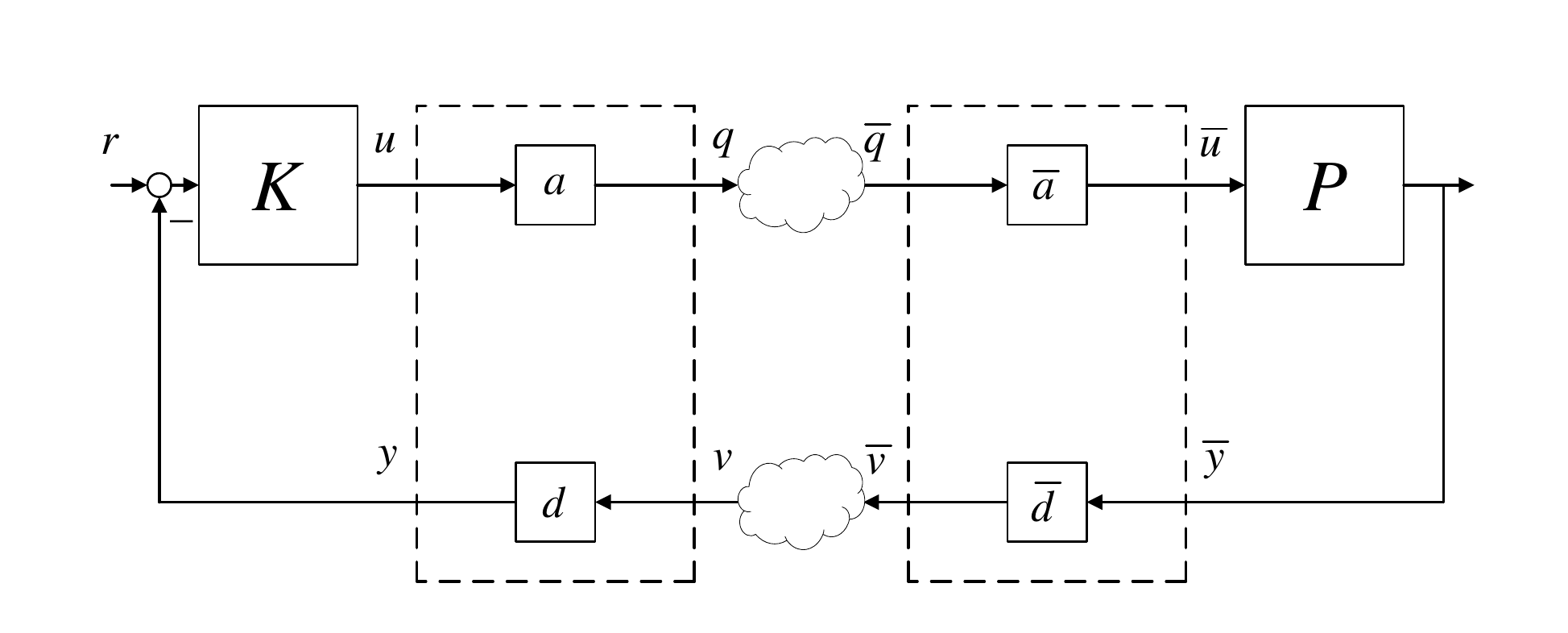}
		\vspace*{-9mm}
		\caption{A networked feedback system with two-way stretching matrix coding.}
		\label{stretch}
	\end{center}
	\vspace*{-3mm}
\end{figure}

Case 3:
\begin{flalign}
M
=\left[
\begin{array}{cc}
a & 0 \\
0 & d \\
\end{array}\right],~ad \neq 0. \nonumber
\end{flalign}
In the case when $ad=1$, $M$ is also known as the two-way squeezing matrix.

The three cases of two-way stretching matrices are easy to understand; they are simply re-scalings of the signals. We now only illustrate case~3 in Fig.~\ref{stretch}. Herein, it is easy to see that $\overline{a} = 1/a$ and $\overline{d} = 1/d$ since
\begin{flalign}
M^{-1}
=\left[
\begin{array}{cc}
\frac{1}{a} & 0 \\
0 & \frac{1}{d} \\
\end{array}\right]. \nonumber
\end{flalign}

As a matter of fact, the two-way stretching matrices reduce to two one-way re-scaling transformations as one-way coding schemes (cf. Fig.~\ref{figureoneway}); we will discuss the differences between two-way coding and one-way coding in more details in the subsequent sections.

\begin{figure}
	\vspace*{-0mm}
	\begin{center}
		\includegraphics [width=0.5\textwidth]{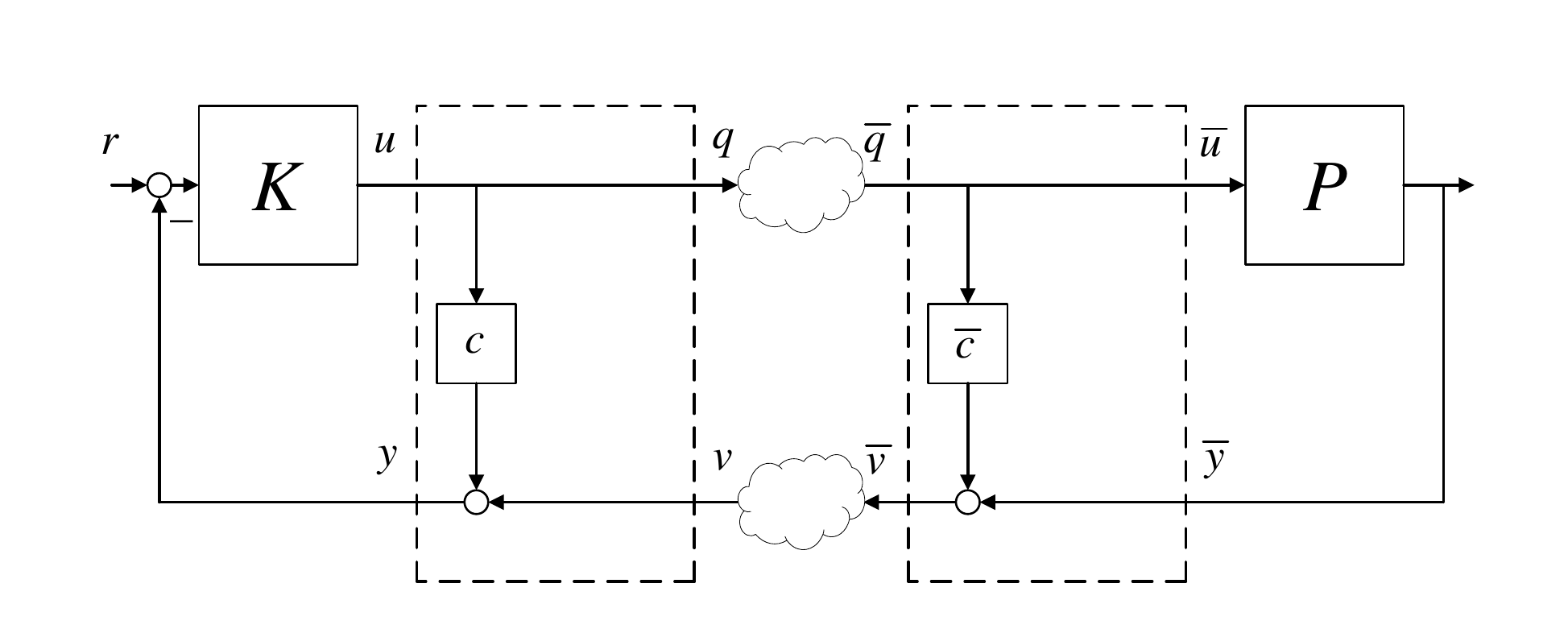}
		\vspace*{-9mm}
		\caption{A networked feedback system with two-way shearing matrix coding: case~1.}
		\label{shear1}
	\end{center}
	\vspace*{-6mm}
\end{figure}

\subsubsection{Two-Way Shearing Matrix} Three cases of the two-way shearing matrices are given below. 

Case 1:
\begin{flalign}
M
=\left[
\begin{array}{cc}
1 & 0 \\
c & 1 \\
\end{array}\right],~
M^{-1}
=\left[
\begin{array}{cc}
1 & 0 \\
-c & 1 \\
\end{array}\right]. \nonumber
\end{flalign}
In this case, we have the illustration given in Fig.~\ref{shear1}, where $\overline{c} = -c$. Simply speaking, the idea is to create a ``parallel system" on the plant side, and compensate for it on the controller side. 

\begin{figure}
	\vspace*{-0mm}
	\begin{center}
		\includegraphics [width=0.5\textwidth]{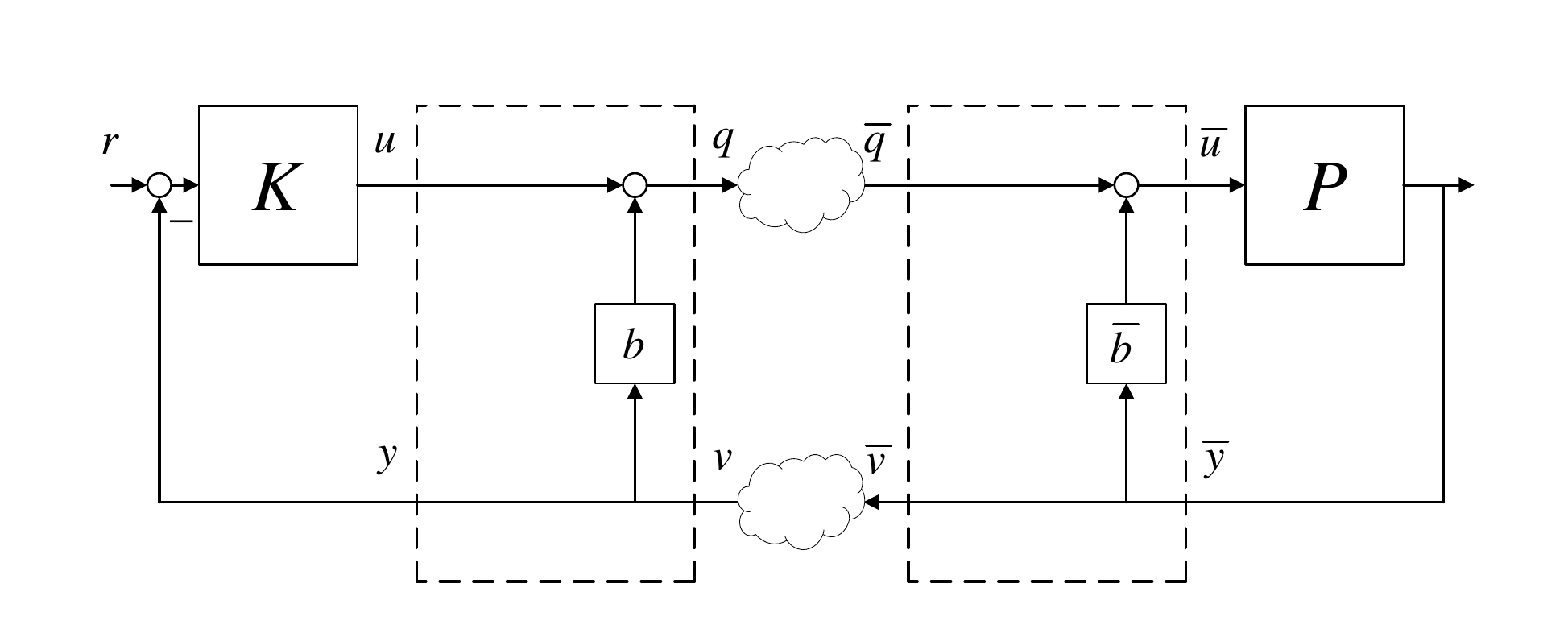}
		\vspace*{-9mm}
		\caption{A networked feedback system with two-way shearing matrix coding: case~2.}
		\label{shear2}
	\end{center}
	\vspace*{-3mm}
\end{figure}

Case 2:
\begin{flalign}
M
=\left[
\begin{array}{cc}
1 & b \\
0 & 1 \\
\end{array}\right],~
M^{-1}
=\left[
\begin{array}{cc}
1 & -b \\
0 & 1 \\
\end{array}\right]. \nonumber
\end{flalign}
In this case, we have the illustration given in Fig.~\ref{shear2}, where $\overline{b} = -b$. The idea is to add a ``local feedback controller" on the plant side, and compensate for it on the controller side.

\begin{figure}
	\vspace*{-0mm}
	\begin{center}
		\includegraphics [width=0.5\textwidth]{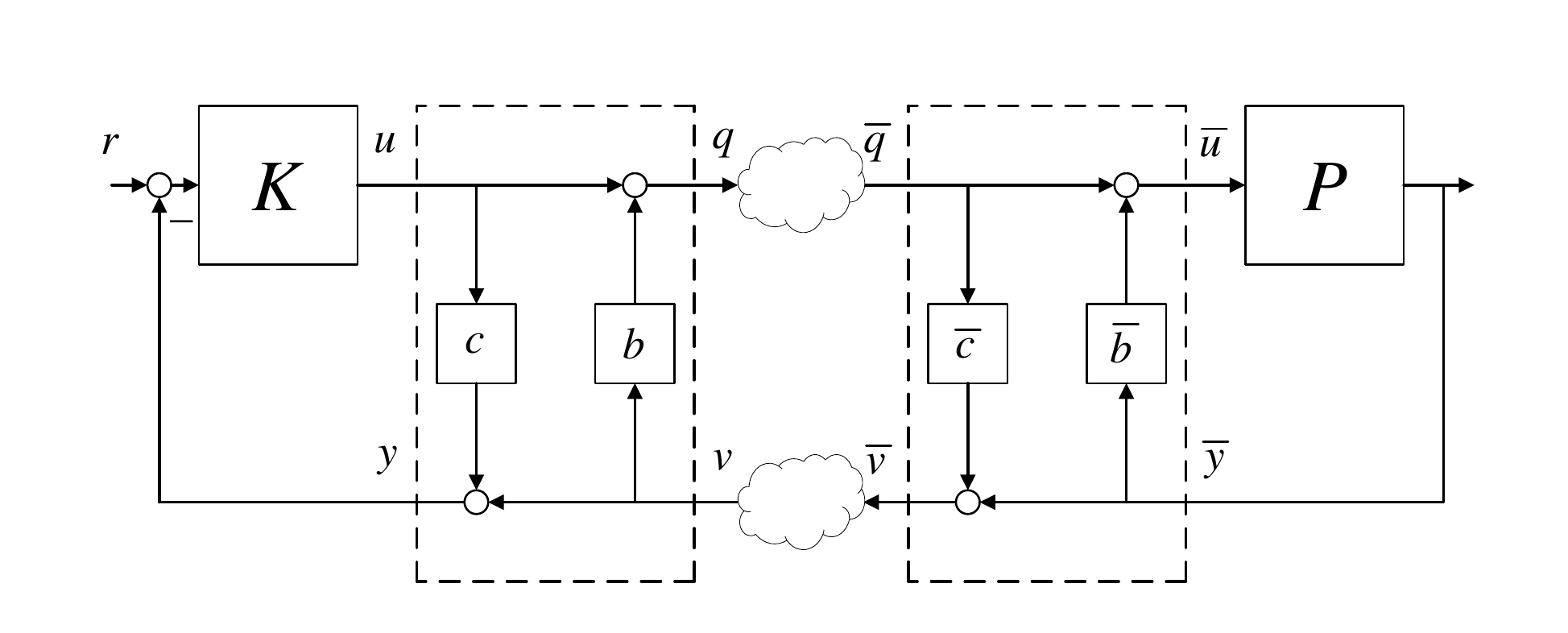}
		\vspace*{-9mm}
		\caption{A networked feedback system with two-way shearing matrix coding: case~3.}
		\label{shear3}
	\end{center}
	\vspace*{-3mm}
\end{figure}

Case 3:
\begin{flalign}
M
=\left[
\begin{array}{cc}
1 & b \\
c & 1 \\
\end{array}\right],~
M^{-1}
=\left[
\begin{array}{cc}
1 & -b \\
-c & 1 \\
\end{array}\right]. \nonumber 
\end{flalign}
Herein, $bc \neq 1$. In this case, we have the illustration given in Fig.~\ref{shear3}, where $\overline{b} = -b$ and $\overline{c} = -c$.

\subsubsection{Two-Way Rotation Matrix} The two-way rotation matrix and its inverse are given by
\begin{flalign}
M
=\left[
\begin{array}{cc}
\cos \theta & \sin \theta \\
-\sin \theta & \cos \theta \\
\end{array}\right],~
M^{-1}
=\left[
\begin{array}{cc}
\cos \theta & - \sin \theta \\
\sin \theta & \cos \theta \\
\end{array}\right]. \nonumber
\end{flalign}
Herein, $ \theta \neq k \pi /2, k= 2m+1, m \in \mathrm{Z}$.

\subsubsection{Scattering Transformation} The scattering transformation is given by \cite{anderson1989bilateral, niemeyer1991stable, hokayem2006bilateral, nuno2011passivity, hirche2012human, hatanaka2015passivity}
\begin{flalign}
\left[
\begin{array}{c}
q \left( t \right)\\
v \left( t \right)\\
\end{array}
\right]
&=\left[
\begin{array}{cc}
\frac{\sqrt{2}}{2} & \frac{\sqrt{2}}{2} \\
-\frac{\sqrt{2}}{2} & \frac{\sqrt{2}}{2} \\
\end{array}\right]
\left[
\begin{array}{cc}
\sqrt{\gamma} & 0 \\
0 & \frac{1}{\sqrt{\gamma}} \\
\end{array}\right]
\left[
\begin{array}{c}
u \left( t \right)\\
y \left( t \right)\\
\end{array}
\right] \nonumber \\
&= \left[
\begin{array}{cc}
\frac{\sqrt{2\gamma}}{2} & \frac{\sqrt{2}}{2\sqrt{\gamma}} \\
-\frac{\sqrt{2\gamma}}{2} & \frac{\sqrt{2}}{2\sqrt{\gamma}} \\
\end{array}\right]
\left[
\begin{array}{c}
u \left( t \right)\\
y \left( t \right)\\
\end{array}
\right]. \nonumber
\end{flalign}
Herein, $0 < \gamma < \infty$.
As a consequence,
\begin{flalign}
\left[
\begin{array}{c}
q \left( t \right)\\
y \left( t \right)\\
\end{array}
\right]
= \left[
\begin{array}{cc}
\sqrt{2\gamma} & 1 \\
\gamma & \sqrt{2\gamma} \\
\end{array}\right]
\left[
\begin{array}{c}
u \left( t \right)\\
v \left( t \right)\\
\end{array}
\right]. \nonumber
\end{flalign}
Correspondingly,
\begin{flalign}
M= \left[
\begin{array}{cc}
\sqrt{2\gamma} & 1 \\
\gamma & \sqrt{2\gamma} \\
\end{array}\right],~
M^{-1}= \left[
\begin{array}{cc}
\sqrt{\frac{2}{\gamma}} & - \frac{1}{\gamma} \\
- 1 & \sqrt{\frac{2}{\gamma}} \\
\end{array}\right]. \nonumber 
\end{flalign}

More generally, the scattering transformation can be extended as \cite{anderson1989bilateral, niemeyer1991stable, hokayem2006bilateral, nuno2011passivity, hirche2012human, hatanaka2015passivity}
\begin{flalign}
\left[
\begin{array}{c}
q \left( t \right)\\
v \left( t \right)\\
\end{array}
\right]
&=\left[
\begin{array}{cc}
\cos \theta & \sin \theta \\
-\sin \theta & \cos \theta \\
\end{array}\right]
\left[
\begin{array}{cc}
\sqrt{\gamma} & 0 \\
0 & \frac{1}{\sqrt{\gamma}} \\
\end{array}\right]
\left[
\begin{array}{c}
u \left( t \right)\\
y \left( t \right)\\
\end{array}
\right] \nonumber \\
&= \left[
\begin{array}{cc}
\sqrt{\gamma} \cos \theta & \frac{1}{\sqrt{\gamma}} \sin \theta \\
-\sqrt{\gamma} \sin \theta & \frac{1}{\sqrt{\gamma}} \cos \theta \\
\end{array}\right]
\left[
\begin{array}{c}
u \left( t \right)\\
y \left( t \right)\\
\end{array}
\right]. \nonumber
\end{flalign}
Herein, $0 < \gamma < \infty$ and $ \theta \neq k \pi /2, k= 2m+1, m \in \mathrm{Z}$. 
As a result,
\begin{flalign}
\left[
\begin{array}{c}
q \left( t \right)\\
y \left( t \right)\\
\end{array}
\right]
= \left[
\begin{array}{cc}
\frac{\sqrt{\gamma}}{\cos \theta} & \tan \theta \\
\gamma \tan \theta & \frac{\sqrt{\gamma}}{\cos \theta}\\
\end{array}\right]
\left[
\begin{array}{c}
u \left( t \right)\\
v \left( t \right)\\
\end{array}
\right], \nonumber
\end{flalign}
and
\begin{flalign}
M= \left[
\begin{array}{cc}
\frac{\sqrt{\gamma}}{\cos \theta} & \tan \theta \\
\gamma \tan \theta & \frac{\sqrt{\gamma}}{\cos \theta}\\
\end{array}\right],~
M^{-1} = \left[
\begin{array}{cc}
\frac{1}{\sqrt{\gamma} \cos \theta} & - \frac{\tan \theta}{\gamma } \\
- \tan \theta & \frac{1}{\sqrt{\gamma} \cos \theta}\\
\end{array}\right]. \nonumber
\end{flalign}

\section{Analysis of LTI Systems with Two-Way Coding} 

In this section, we analyze in particular LTI feedback control systems. Consider the SISO feedback system with two-way coding depicted in Fig.~\ref{injection}. Assume that herein the controller $K$ and plant $P$ are LTI with transfer functions $K \left( s \right)$ and $P \left( s \right)$, respectively. In addition, let $r \left( t \right)$, $u \left( t \right)$, $\overline{u} \left( t \right)$, $y \left( t \right)$, $\overline{y} \left( t \right)$, $q \left( t \right)$, $ \overline{q} \left( t \right)$, $v \left( t \right)$, $\overline{v} \left( t \right) \in \mathbb{R}$. Meanwhile, suppose that injection (additive) attacks $w \left( t \right) \in \mathbb{R}$ and $z \left( t \right) \in \mathbb{R}$ exist in the forward path and feedback path of the control systems, respectively. Let $R \left( s \right)$, $U \left( s \right)$, $ \overline{U} \left( s \right)$, $Y \left( s \right)$, $\overline{Y} \left( s \right)$, $Q \left( s \right)$, $\overline{Q} \left( s \right)$, $V \left( s \right)$, $\overline{V} \left( s \right)$, $W \left( s \right)$, $Z \left( s \right)$ represent the Laplace transforms, assuming that they exist, of the signals $r \left( t \right)$, $u \left( t \right)$, $\overline{u} \left( t \right)$, $y \left( t \right)$, $\overline{y} \left( t \right)$, $q \left( t \right)$, $\overline{q} \left( t \right)$, $v \left( t \right)$, $\overline{v} \left( t \right)$, $ w \left( t \right)$, $z \left( t \right)$. From now on, we assume that all the transfer functions of the systems are with zero initial conditions, unless otherwise specified.

\begin{figure}
	\vspace*{-0mm}
	\begin{center}
		\includegraphics [width=0.5\textwidth]{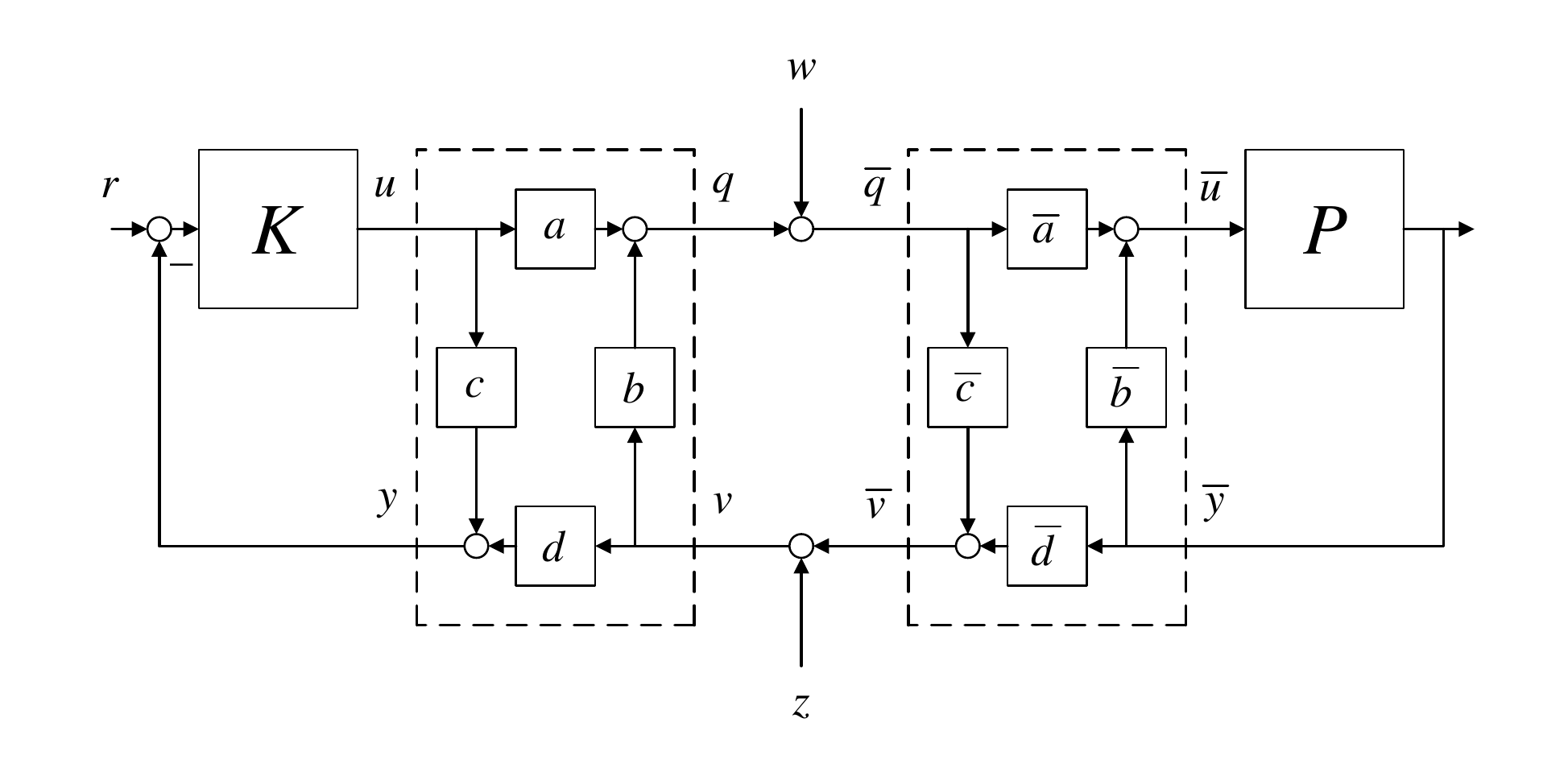}
		\vspace*{-9mm}
		\caption{A feedback system with two-way coding under injection attacks.}
		\label{injection}
	\end{center}
	\vspace*{-3mm}
\end{figure}

We first provide expressions for the Laplace transforms of the real plant output $\overline{y} \left( t \right)$ and the plant output $y \left( t \right)$ as seen on the controller side, given reference $r \left( t \right)$ and under injection attacks $w \left( t \right)$ and $z \left( t \right)$.

\begin{theorem} \label{foundation}
	Consider the SISO feedback system with two-way coding under injection attacks depicted in Fig.~\ref{injection}.
	Assume that controller $K$ and plant $P$ are LTI with transfer functions $K \left( s \right)$ and $P \left( s \right)$, respectively, and that the closed-loop system is stable.
	Then, 
	\begin{flalign} \label{injection1}
	\overline{Y} \left( s \right)
	&= \frac{K \left( s \right) P \left( s \right)}{1+ K \left( s \right) P \left( s \right)} R \left( s \right)  + \frac{ a^{-1}\left[ 1 + c K \left( s \right) \right] P \left( s \right)}{1+ K \left( s \right) P \left( s \right)} W \left( s \right)
	\nonumber \\
	&\ \ \ \  + \frac{ a^{-1} \left[ b - \left( ad - bc \right) K \left( s \right) \right] P \left( s \right)}{1+ K \left( s \right) P \left( s \right)} Z \left( s \right),
	\end{flalign}
	and
	\begin{flalign} \label{injection2}
	Y \left( s \right) 
	&= \frac{K \left( s \right) P \left( s \right)}{1+ K \left( s \right) P \left( s \right)} R \left( s \right) +  \frac{a^{-1} \left[ P \left( s \right) - c \right]}{1+ K \left( s \right) P \left( s \right)}  W \left( s \right) \nonumber \\
	&\ \ \ \  +  \frac{a^{-1} \left[ ad - bc + b P \left( s \right) \right]}{1+ K \left( s \right) P \left( s \right)} Z \left( s \right)
	. 
	\end{flalign}
\end{theorem}

\vspace*{0mm}

\begin{proof}
	Since 
	\begin{flalign}
	\overline{Y} \left( s \right) 
	= P \left( s \right) \overline{U} \left( s \right), \nonumber 
	\end{flalign} 
	and 
	\begin{flalign}
	\overline{U} \left( s \right) 
	= \overline{b} \overline{Y} \left( s \right) + \overline{a} \overline{Q} \left( s \right), \nonumber 
	\end{flalign}
	we have
	\begin{flalign}
	\overline{U} \left( s \right) 
	= \overline{b} P \left( s \right) \overline{U} \left( s \right) + \overline{a} \overline{Q} \left( s \right), \nonumber 
	\end{flalign}
	and thus
	\begin{flalign}
	\overline{U} \left( s \right) 
	= \frac{\overline{a} \overline{Q} \left( s \right)}{1 - \overline{b} P \left( s \right)}. \nonumber 
	\end{flalign}
	Correspondingly,
	\begin{flalign}
	\overline{Y} \left( s \right) 
	= P \left( s \right) \overline{U} \left( s \right)
	= \frac{\overline{a} P \left( s \right) \overline{Q} \left( s \right)}{1 - \overline{b} P \left( s \right)}. \nonumber 
	\end{flalign}
	As a consequence,
	\begin{flalign} \label{equalP}
	\overline{V} \left( s \right) 
	&= \overline{d} \overline{Y} \left( s \right) + \overline{c} \overline{Q} \left( s \right)
	= \frac{\overline{a} \overline{d} P \left( s \right) \overline{Q} \left( s \right)}{1 - \overline{b} P \left( s \right)} + \overline{c} \overline{Q} \left( s \right) \nonumber \\
	&= \left[ \frac{\overline{a} \overline{d} P \left( s \right)}{1 - \overline{b} P \left( s \right)} + \overline{c} \right] \overline{Q} \left( s \right) 
	= \frac{\left( \overline{a} \overline{d} - \overline{b} \overline{c} \right) P \left( s \right) + \overline{c}}{1 - \overline{b} P \left( s \right)} \overline{Q} \left( s \right)  \nonumber \\
	&= \frac{P \left( s \right) - c}{ad - bc + b P \left( s \right)} \overline{Q} \left( s \right) 
	. 
	\end{flalign}
	Similarly, since 
	\begin{flalign}
	U \left( s \right) 
	= K \left( s \right) \left[ R \left( s \right) - Y \left( s \right) \right], \nonumber 
	\end{flalign} 
	and 
	\begin{flalign}
	Y \left( s \right) 
	= c U \left( s \right) + d V \left( s \right), \nonumber 
	\end{flalign}
	we have
	\begin{flalign}
	Y \left( s \right) 
	= c K \left( s \right) \left[ R \left( s \right) - Y \left( s \right) \right] + d V \left( s \right), \nonumber 
	\end{flalign}
	and hence
	\begin{flalign}
	Y \left( s \right) 
	= \frac{c K \left( s \right) R \left( s \right) + d V \left( s \right) }{1 + c K \left( s \right)}. \nonumber 
	\end{flalign}
	In addition,
	\begin{flalign}
	U \left( s \right) 
	&= K \left( s \right) \left[ R \left( s \right) - Y \left( s \right) \right] \nonumber \\
	&= K \left( s \right) R \left( s \right) - \frac{ K \left( s \right) \left[ c K \left( s \right) R \left( s \right) + d V \left( s \right) \right] }{1 + c K \left( s \right)} \nonumber \\
	&= \frac{ K \left( s \right) \left[ R \left( s \right) - d V \left( s \right) \right]}{1 + c K \left( s \right)}. \nonumber
	\end{flalign}
	Thus,
	\begin{flalign} \label{equalK}
	Q \left( s \right) 
	&= a U \left( s \right) + b V \left( s \right) \nonumber \\
	&= \frac{ a K \left( s \right)}{1 + c K \left( s \right)}  R \left( s \right) - \frac{ a d K \left( s \right)}{1 + c K \left( s \right)} V \left( s \right) + b V \left( s \right) \nonumber \\
	&= \frac{ a K \left( s \right)}{1 + c K \left( s \right)} R \left( s \right) + \frac{ b - \left( a d - b c \right) K \left( s \right)}{1 + c K \left( s \right)} V \left( s \right) \nonumber \\
	&= \frac{ b - \left( a d - b c \right) K \left( s \right)}{1 + c K \left( s \right)} \left[ \frac{ a K \left( s \right)}{b - \left( a d - b c \right) K \left( s \right)} \right] R \left( s \right) \nonumber \\ 
	&\ \ \ \  + \frac{ b - \left( a d - b c \right) K \left( s \right)}{1 + c K \left( s \right)} V \left( s \right). 
	\end{flalign}
	Using \eqref{equalP} and \eqref{equalK}, while noting that
	\begin{flalign}
	\overline{Q} \left( s \right) =  Q \left( s \right) + W \left( s \right), \nonumber 
	\end{flalign}
	and
	\begin{flalign}
	V \left( s \right) =  \overline{V} \left( s \right) + Z \left( s \right), \nonumber 
	\end{flalign}
	we may then obtain that
	\begin{flalign}
	\overline{Q} \left( s \right) 
	&=  Q \left( s \right) + W \left( s \right) \nonumber \\
	&=  W \left( s \right) +  \frac{ a K \left( s \right)}{1 + c K \left( s \right)}  R \left( s \right) + \frac{ b - \left( a d - b c \right) K \left( s \right)}{1 + c K \left( s \right)} V \left( s \right) \nonumber \\
	&= W \left( s \right) +  \frac{ a K \left( s \right)}{1 + c K \left( s \right)} R \left( s \right)  \nonumber \\
	&\ \ \ \  + \frac{ b - \left( a d - b c \right) K \left( s \right)}{1 + c K \left( s \right)} \left[ \overline{V} \left( s \right) + Z \left( s \right) \right] \nonumber \\
	&=  W \left( s \right) + \frac{ a K \left( s \right)}{1 + c K \left( s \right)}  R \left( s \right) \nonumber \\
	&\ \ \ \  + \frac{ b - \left( a d - b c \right) K \left( s \right)}{1 + c K \left( s \right)} \left[ \frac{P \left( s \right) - c}{ad - bc + b P \left( s \right)} \right] \overline{Q} \left( s \right) \nonumber \\
	&\ \ \ \  +  \frac{ b - \left( a d - b c \right) K \left( s \right)}{1 + c K \left( s \right)} Z \left( s \right). \nonumber 
	\end{flalign}
	Hence,
	\begin{flalign}
	&\left\{ 1 - \frac{ b - \left( a d - b c \right) K \left( s \right)}{1 + c K \left( s \right)} \left[ \frac{P \left( s \right) - c}{ad - bc + b P \left( s \right)} \right] \right\}\overline{Q} \left( s \right)  \nonumber \\
	&\ \ \ \ =  W \left( s \right) + \frac{ a K \left( s \right)}{1 + c K \left( s \right)}  R \left( s \right) +  \frac{ b - \left( a d - b c \right) K \left( s \right)}{1 + c K \left( s \right)} Z \left( s \right). \nonumber 
	\end{flalign}
	On the other hand,
	\begin{flalign} 
	&1- \frac{ b - \left( a d - b c \right) K \left( s \right)}{1 + c K \left( s \right)} \left[ \frac{P \left( s \right) - c}{ad - bc + b P \left( s \right)} \right] \nonumber \\
	&\ \ \ \ = \frac{\left[ 1 + c K \left( s \right) \right] \left[ ad - bc + b P \left( s \right) \right]}{\left[ 1 + c K \left( s \right) \right] \left[ ad - bc + b P \left( s \right) \right]} \nonumber \\
	&\ \ \ \ \ \ \ \ - \frac{\left[ b - \left( a d - b c \right) K \left( s \right)\right] \left[  P \left( s \right) - c \right]}{\left[ 1 + c K \left( s \right) \right] \left[ ad - bc + b P \left( s \right) \right]} \nonumber \\
	&\ \ \ \ = \frac{ad\left[ 1 + K \left( s \right) P \left( s \right) \right]}{\left[ 1 + c K \left( s \right) \right] \left[ ad - bc + b P \left( s \right) \right]}
	. \nonumber
	\end{flalign}
	As a result,
	\begin{flalign}
	&\overline{Q} \left( s \right)
	= \frac{d^{-1}K \left( s \right) \left[ ad - bc + b P \left( s \right) \right]}{1+ K \left( s \right) P \left( s \right)} R \left( s \right) \nonumber \\
	&\ \ \ \  + \frac{a^{-1} d^{-1} \left[ 1 + c K \left( s \right) \right] \left[ ad - bc + b P \left( s \right) \right]}{1+ K \left( s \right) P \left( s \right)} W \left( s \right)
	\nonumber \\
	&\ \ \ \  + \frac{ a^{-1} d^{-1} \left[ b - \left( ad - bc \right) K \left( s \right) \right] \left[ ad - bc + b P \left( s \right) \right]}{1+ K \left( s \right) P \left( s \right)} Z \left( s \right). \nonumber
	\end{flalign}
	Thus,
	\begin{flalign}
	\overline{Y} \left( s \right) 
	&= \frac{\overline{a} P \left( s \right) \overline{Q} \left( s \right)}{1 - \overline{b} P \left( s \right)} 
	= \frac{d P \left( s \right) \overline{Q} \left( s \right)}{ad - bc + b P \left( s \right)} \nonumber \\
	&= \frac{K \left( s \right) P \left( s \right)}{1+ K \left( s \right) P \left( s \right)} R \left( s \right)  + \frac{ a^{-1}\left[ 1 + c K \left( s \right) \right] P \left( s \right)}{1+ K \left( s \right) P \left( s \right)} W \left( s \right)
	\nonumber \\
	&\ \ \ \  + \frac{ a^{-1} \left[ b - \left( ad - bc \right) K \left( s \right) \right] P \left( s \right)}{1+ K \left( s \right) P \left( s \right)} Z \left( s \right). \nonumber
	\end{flalign}
	Similarly, we have
	\begin{flalign}
	V \left( s \right) 
	&=  \overline{V} \left( s \right) + Z \left( s \right)
	= \frac{P \left( s \right) - c}{ad - bc + b P \left( s \right)} \overline{Q} \left( s \right) + Z \left( s \right) \nonumber \\
	&= \frac{P \left( s \right) - c}{ad - bc + b P \left( s \right)} \left[Q \left( s \right) + W \left( s \right)\right] + Z \left( s \right) \nonumber \\
	&= \frac{P \left( s \right) - c}{ad - bc + b P \left( s \right)} \left[ \frac{ a K \left( s \right)}{1 + c K \left( s \right)} \right] R \left( s \right) \nonumber \\
	&\ \ \ \  + \frac{P \left( s \right) - c}{ad - bc + b P \left( s \right)} \left[ \frac{ b - \left( a d - b c \right) K \left( s \right)}{1 + c K \left( s \right)} \right] V \left( s \right) \nonumber \\
	&\ \ \ \  + \frac{P \left( s \right) - c}{ad - bc + b P \left( s \right)}  W \left( s \right) + Z \left( s \right)
	, \nonumber 
	\end{flalign}
	and
	\begin{flalign}
	V \left( s \right) 
	&= \frac{d^{-1} K \left( s \right) \left[P \left( s \right) - c \right]}{1+ K \left( s \right) P \left( s \right)} R \left( s \right) \nonumber \\
	&\ \ \ \  +  \frac{a^{-1} d^{-1} \left[ 1 + c K \left( s \right) \right] \left[ P \left( s \right) - c \right]}{1+ K \left( s \right) P \left( s \right)}  W \left( s \right) \nonumber \\
	&\ \ \ \  +  \frac{a^{-1} d^{-1} \left[ 1 + c K \left( s \right) \right] \left[ ad - bc + b P \left( s \right) \right]}{1+ K \left( s \right) P \left( s \right)} Z \left( s \right)
	. \nonumber 
	\end{flalign}
	Consequently,
	\begin{flalign}
	Y \left( s \right) 
	&= \frac{c K \left( s \right) R \left( s \right) + d V \left( s \right) }{1 + c K \left( s \right)} \nonumber \\
	&= \frac{K \left( s \right) P \left( s \right)}{1+ K \left( s \right) P \left( s \right)} R \left( s \right) +  \frac{a^{-1} \left[ P \left( s \right) - c \right]}{1+ K \left( s \right) P \left( s \right)}  W \left( s \right) \nonumber \\
	&\ \ \ \  +  \frac{a^{-1} \left[ ad - bc + b P \left( s \right) \right]}{1+ K \left( s \right) P \left( s \right)} Z \left( s \right)
	. \nonumber 
	\end{flalign}
	This completes the proof.
\end{proof}

\begin{figure}
	\vspace*{-0mm}
	\begin{center}
		\includegraphics [width=0.5\textwidth]{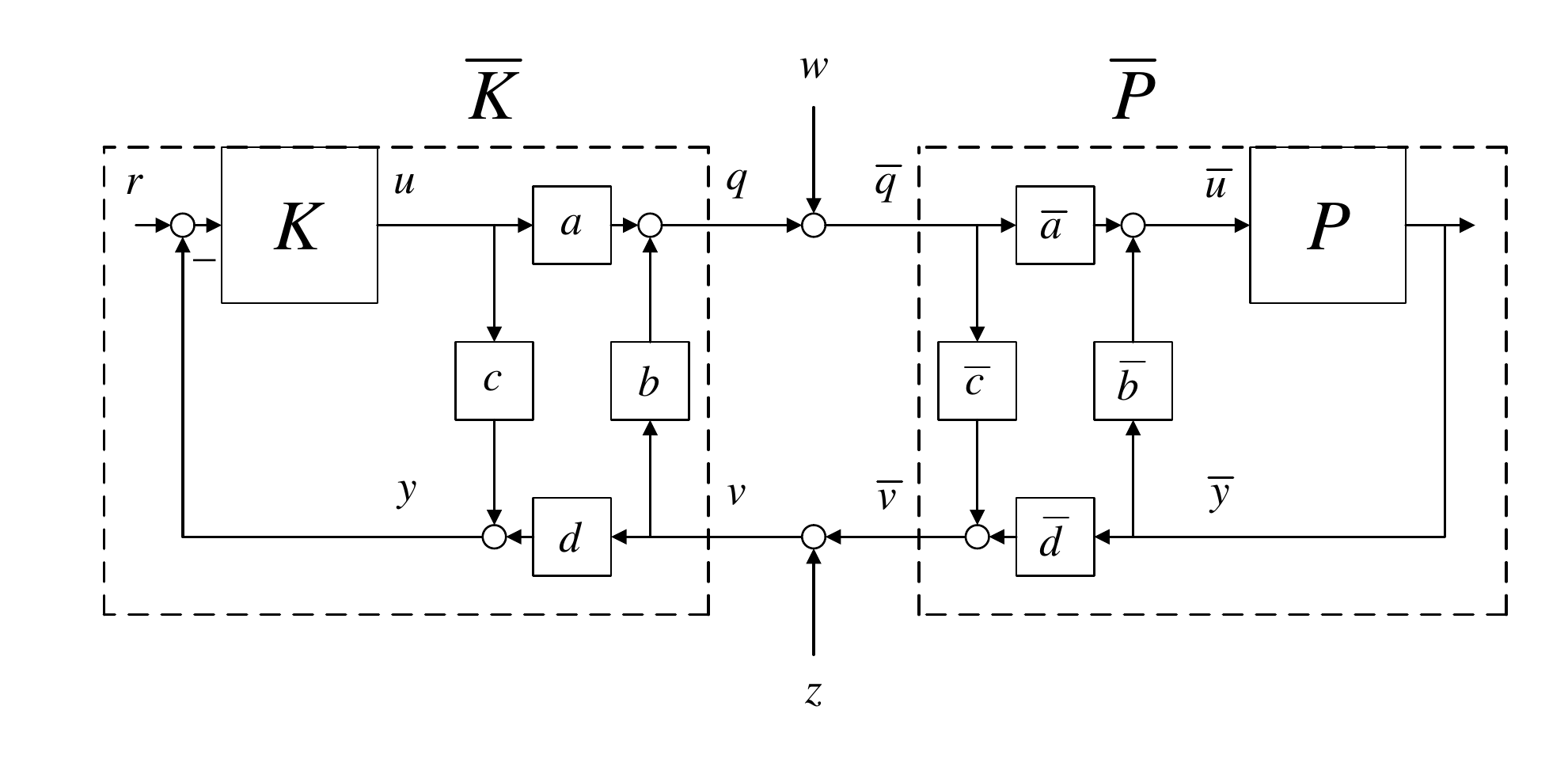}
		\vspace*{-9mm}
		\caption{A feedback system with two-way coding: from the viewpoint of the attacker.}
		\label{injection_e1}
	\end{center}
	\vspace*{-3mm}
\end{figure}

\begin{figure}
	\vspace*{-0mm}
	\begin{center}
		\includegraphics [width=0.25\textwidth]{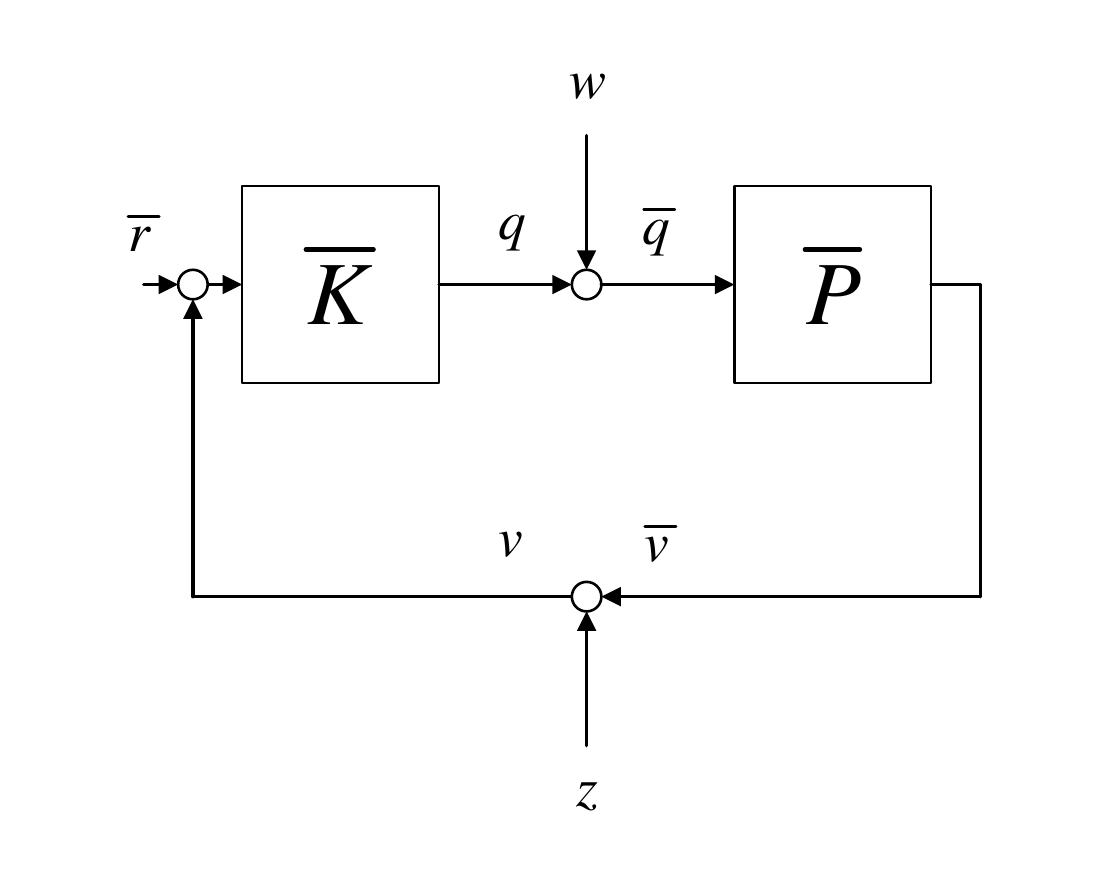}
		\vspace*{-6mm}
		\caption{A feedback system with two-way coding: the equivalent system from the attacker's viewpoint.}
		\label{injection_e2}
	\end{center}
	\vspace*{-3mm}
\end{figure}


\begin{figure}
	\vspace*{-0mm}
	\begin{center}
		\includegraphics [width=0.25\textwidth]{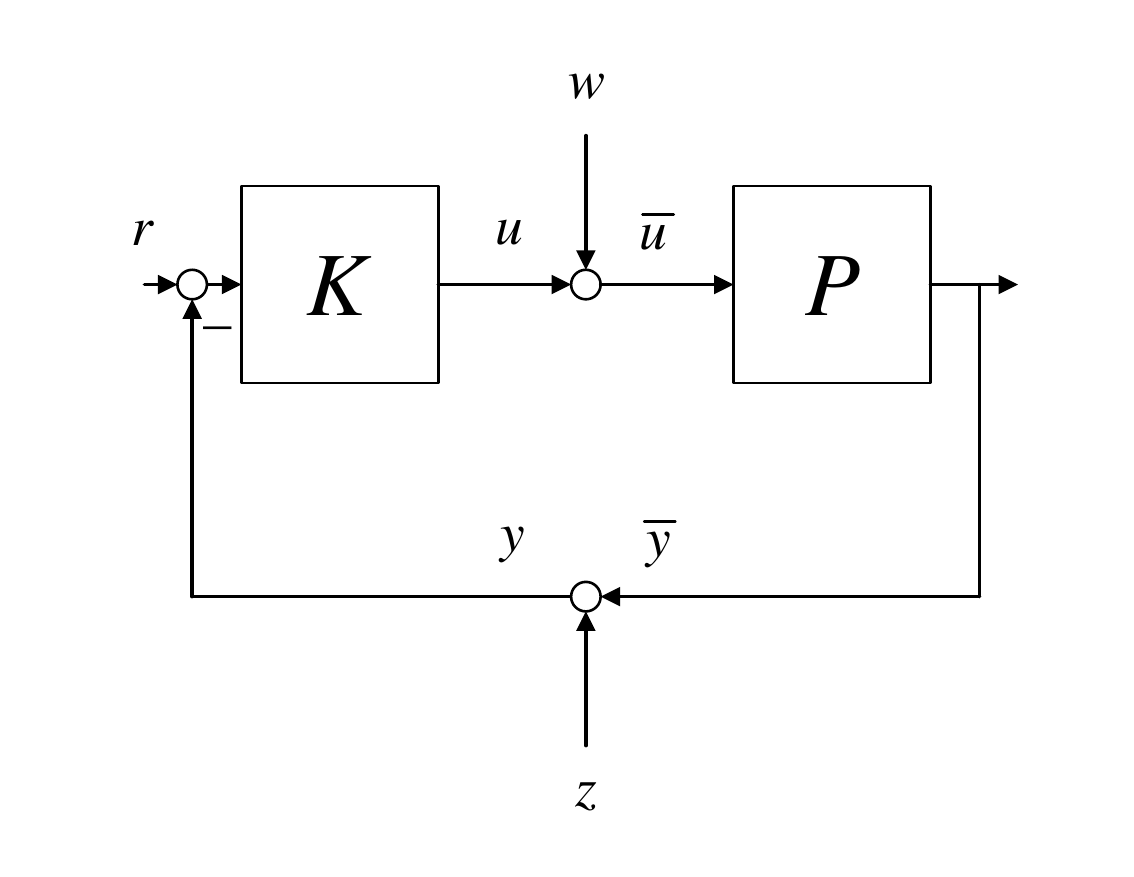}
		\vspace*{-6mm}
		\caption{The original feedback system without two-way coding.}
		\label{figure3_2}
	\end{center}
	\vspace*{-3mm}
\end{figure}


Note that based on Theorem~\ref{foundation}, Laplace transforms of all the signals flowing in the feedback system can be obtained. For instance, since $\overline{Y} \left( s \right) = P \left( s \right) \overline{U} \left( s \right)$, it follows from \eqref{injection1} that plant input $\overline{U} \left( s \right)$ is given by
\begin{flalign} 
\overline{U} \left( s \right)
&= \frac{K \left( s \right)}{1+ K \left( s \right) P \left( s \right)} R \left( s \right) 
+ \frac{ a^{-1}\left[ 1 + c K \left( s \right) \right]}{1+ K \left( s \right) P \left( s \right)} W \left( s \right)
\nonumber \\
&\ \ \ \  + \frac{ a^{-1} \left[ b - \left( ad - bc \right) K \left( s \right) \right]}{1+ K \left( s \right) P \left( s \right)} Z \left( s \right).
\end{flalign}

We now investigate the implications of Theorem~\ref{foundation}. It is clear that from the perspective of the reference, the transfer function from reference $R \left( s \right)$ to plant output $\overline{Y} \left( s \right)$, found as
\begin{flalign} \label{original}
\frac{K \left( s \right) P \left( s \right)}{1+ K \left( s \right) P \left( s \right)},
\end{flalign}
stays exactly the same as in the original system depicted in Fig.~\ref{figure3_2} where two-way coding does not exist; therein, the transfer function from reference $R \left( s \right)$ to plant output $\overline{Y} \left( s \right)$ is also given by \eqref{original}. As such, the controller $K \left( s \right)$ may be designed regardless of two-way coding.
Meanwhile, to the attacker, the feedback system behaves differently from the original system because of the presence of two-way coding, as will be shown in the following corollary.

\begin{corollary}
	From the viewpoint of the attacker (see Fig.~\ref{injection_e1}), the feedback system is equivalent to that of Fig.~\ref{injection_e2}, where the transfer function of the equivalent plant $\overline{P}$ is given by
	\begin{flalign} \label{equalP1}
	\overline{P} \left( s \right)
	= \frac{P \left( s \right) - c}{ad - bc + b P \left( s \right)},
	\end{flalign} 
	while that of the equivalent controller $\overline{K}$ is found as
	\begin{flalign} \label{equalK1}
	\overline{K} \left( s \right)
	= \frac{b - \left( ad - bc \right) K \left( s \right)}{1 + c K \left( s \right)}.
	\end{flalign} 
	In addition, the Laplace transform of the equivalent reference signal $\overline{r}$ is 
	\begin{flalign} \label{equalR}
	\overline{R} \left( s \right)
	= \frac{ a K \left( s \right)}{b - \left( a d  - b c \right) K \left( s \right)} R \left( s \right).
	\end{flalign}
\end{corollary}


\vspace*{0mm}

\begin{proof}
	Equations \eqref{equalP1}, \eqref{equalK1}, and \eqref{equalR} follow directly from \eqref{equalP} and \eqref{equalK}.
\end{proof}

%

Clearly, the presence of two-way coding will distort the attacker's view of the control system, making the properties of the plant, controller, and reference all seemingly different from those of the original system. This distorted perspective will assist in defending the system against attacks that are designed based on the system models, as will be seen shortly in the next section.

\section{Attack Detection and Correction}

In this section, we examine how the presence of two-way coding can play a critical role in the defense against injection attacks in LTI systems. In general, the distorted perspective on the attacker side as a result of two-way coding will enable detecting the attacks or restricting what the attacker can do or even correcting the attack effect, depending on the attacker's knowledge of the system. In the particular case of zero-dynamics attacks, it is seen that the attacks will be detected if designed according to the original plant, while the attack effect will be corrected in steady state if the attacks are to be designed with respect to the equivalent plant as seen by the attacker.


Before we proceed, we first prove the following result. Consider still the SISO feedback system depicted in Fig.~\ref{injection}. Without physically changing $P \left( s \right)$, we can use two-way coding to make the zeros and/or poles of the equivalent plant $\overline{P} \left( s \right) $, as seen by the attacker, all different from those of the original plant $P \left( s \right)$.

\begin{theorem} \label{zero}
	Let
	\begin{flalign}
	P \left( s \right) = \frac{m_{P} \left( s \right)}{n_{P} \left( s \right)},
	\end{flalign}
	where $m_{P} \left( s \right)$ and $n_{P} \left( s \right)$ denote the numerator and denominator polynomials of $P \left( s \right)$, respectively. Suppose that $m_{P} \left( s \right)$ and $n_{P} \left( s \right)$ are coprime. 
	\begin{itemize}
		\item The zeros of $\overline{P} \left( s \right)$ are given by the roots of 
		\begin{flalign}
		m_{P} \left( s \right) - cn_{P} \left( s \right) = 0.
		\end{flalign}
		In addition, if $c \neq 0$, then the zeros of $\overline{P} \left( s \right)$ are all different from those of $P \left( s \right)$.
		\item The poles of $\overline{P} \left( s \right)$ are given by the roots of 
		\begin{flalign}
		\left(ad - bc\right)n_{P} \left( s \right) + b m_{P} \left( s \right) = 0.
		\end{flalign}
		In addition, if $b \neq 0$, then the poles of $\overline{P} \left( s \right)$ are all different from those of $P \left( s \right)$.
	\end{itemize}
\end{theorem}

\begin{proof}
	It is clear that
	\begin{flalign}
	\overline{P} \left( s \right) 
	=\frac{P \left( s \right) - c}{ad - bc + b P \left( s \right)}
	= \frac{m_{P} \left( s \right) - cn_{P} \left( s \right)}{\left(ad - bc\right)n_{P} \left( s \right) + b m_{P} \left( s \right)}. \nonumber
	\end{flalign}
	
	Note that the zeros of $P \left( s \right)$ are given by the roots of $m_{P} \left( s \right) = 0$. Let $z_{i}$ be a zero of $P \left( s \right)$. Hence, $m_{P} \left( z_{i} \right) = 0$. Then, $z_{i}$ cannot be a zero of $\overline{P} \left( s \right)$, otherwise this will lead to $n_{P} \left( z_{i} \right) = m_{P} \left( z_{i} \right) /c= 0$ and thus $z_{i}$ will be a pole of $P \left( s \right)$, which contradicts the fact that $m_{P} \left( s \right)$ and $n_{P} \left( s \right)$ are coprime. 
	
	Similarly, note that the poles of $P \left( s \right)$ are given by the roots of $n_{P} \left( s \right) = 0$. Let $p_{i}$ be a pole of $P \left( s \right)$. Therefore, $n_{P} \left( p_{i} \right) = 0$. Then, $p_{i}$ cannot be a pole of $\overline{P} \left( s \right)$, otherwise this will lead to $m_{P} \left( p_{i} \right) = \left(ad - bc\right) n_{P} \left( p_{i} \right)/b = 0$ and thus $p_{i}$ will be a zero of $P \left( s \right)$, which contradicts the fact that $m_{P} \left( s \right)$ and $n_{P} \left( s \right)$ are coprime.
\end{proof}

Note that herein the conditions $c \neq 0$ and/or $b \neq 0$ are essential.
Similarly to Theorem~\ref{zero}, it can also be shown that the zeros of $\overline{K} \left( s \right)$ are all different from those of $K \left( s \right)$ when $b \neq 0$, while the poles of $\overline{K} \left( s \right)$ are all different from those of $K \left( s \right)$ when $c \neq 0$. 


We now remark on some fundamental differences between two-way coding and one-way coding. It is clear that two-way coding reduces to two one-way coding schemes when $b=c=0$ (as in the case of two-way stretching matrix; see Section~\ref{specialtwoway}), and correspondingly, 
\begin{flalign} 
\overline{P} \left( s \right) = \frac{P \left( s \right)}{ad},~\overline{K} \left( s \right) = ad K \left( s \right).
\end{flalign} 
It is clear that the zeros and poles of $\overline{P} \left( s \right)$ and $\overline{K} \left( s \right)$ are exactly the same as those of $P \left( s \right)$ and $K \left( s \right)$.  
In other words, one-way coding will not change the zeros nor poles of the plant nor the controller. Indeed, similar results hold for multiple-input multiple-output (MIMO) systems as well.
It is also worth mentioning that even with dynamic one-way coding schemes, since cancellations between unstable poles and nonminimum-phase zeros should always be avoided to prevent possible internal instability, the nonminimum-phase zeros and unstable poles of the original plant and controller cannot be eliminated.

We next show that the presence of two-way coding not only can make the zeros and/or poles of the equivalent plant $\overline{P} \left( s \right)$ all different from those of the original plant $P \left( s \right)$, but also, under some additional conditions, may render $\overline{P} \left( s \right)$ stable and/or minimum-phase.
In fact, similar results hold for the pair of the equivalent controller $\overline{K} \left( s \right) $ and the original controller $K \left( s \right)$ as well.

\begin{theorem} \label{outputfeedback}
	Suppose that a minimal realization of the plant $P \left( s \right)$ is given by
	\begin{flalign} \label{minimal}
	\left\{ \begin{array}{rcl}
	\dot{x} \left( t \right) & = & A x \left( t \right) + B \overline{u} \left( t \right),\\
	\overline{y} \left( t \right) & = & C x \left( t \right) + D \overline{u} \left( t \right).
	\end{array} \right.
	\end{flalign}
	If the plant is stabilizable by static output feedback \cite{syrmos1997static} described as
	\begin{flalign}
	\overline{u} \left( t \right) = F \overline{y} \left( t \right), 
	\end{flalign}
	where $F \in \mathbb{R}$, then all the poles of $\overline{P} \left( s \right)$ can be made stable, while all the zeros of $\overline{P} \left( s \right)$ can be made minimum-phase. 
\end{theorem}

\begin{proof}
	Since \eqref{minimal} is a minimal realization of $P \left( s \right)$, we have
	\begin{flalign}
	P \left( s \right) = \frac{m_{P} \left( s \right)}{n_{P} \left( s \right)} = C \left(sI - A\right)^{-1} B + D. \nonumber
	\end{flalign}
	Meanwhile, since $P$ is stabilizable by static output feedback, there exists a non-zero constant $F_{1}$ such that
	\begin{flalign}
	\frac{1}{1 + F_{1} \left[ C \left(sI - A\right)^{-1} B + D \right]} 
	&= \frac{1}{1 + F_{1} \left[ \frac{m_{P} \left( s \right)}{n_{P} \left( s \right)} \right]} \nonumber \\
	&= \frac{n_{P} \left( s \right)}{n_{P} \left( s \right) + F_{1} m_{P} \left( s \right)} \nonumber
	\end{flalign}
	is stable, i.e., all its poles are stable. In other words, all the roots of $n_{P} \left( s \right) + F_{1} m_{P} \left( s \right)$ are with negative real parts. Meanwhile, note that
	\begin{flalign}
	\overline{P} \left( s \right) 
	=\frac{P \left( s \right) - c}{ad - bc + b P \left( s \right)}
	= \frac{m_{P} \left( s \right) - cn_{P} \left( s \right)}{\left(ad - bc\right)n_{P} \left( s \right) + b m_{P} \left( s \right)}. \nonumber
	\end{flalign}
	As such, when $b/\left(ad-bc\right) = F_{1}$, 
	\begin{flalign}
	\left(ad - bc\right)n_{P} \left( s \right) + b m_{P} \left( s \right)
	= \left(ad - bc\right) \left[ n_{P} \left( s \right) + F_{1} m_{P} \left( s \right) \right], \nonumber
	\end{flalign}
	and hence all its roots are with negative real parts, i.e., all the poles of $\overline{P} \left( s \right)$ are stable. Similarly, when $c = -1/F_{2}$, where $F_{2}$ is a stabilizing, non-zero static output feedback control gain, we have
	\begin{flalign}
	m_{P} \left( s \right) - cn_{P} \left( s \right) 
	= -c \left[ n_{P} \left( s \right) + F_{2} m_{P} \left( s \right) \right], \nonumber
	\end{flalign} 
	and thus all its roots are with negative real parts, i.e., all the zeros of $\overline{P} \left( s \right)$ are minimum-phase.
\end{proof}

From the proof, it can be seen that it is possible to make all the poles of $\overline{P} \left( s \right)$ stable and all the zeros of $\overline{P} \left( s \right)$ minimum-phase simultaneously, as long as $F_{1} \neq F_{2}$. 

It is also worth mentioning that herein we only require the plant to be stabilizable by static output feedback $F$, which is used merely for the purpose of deciding the parameters of two-way coding, but the controller $K \left(s\right)$ is not necessarily chosen among such static controllers; stated alternatively, the controller $K \left(s\right)$ are not further restricted.

\subsection{Zero-Dynamics Attacks} \label{zerodynamicsattacks}

We next examine the implications of Theorem~\ref{zero} and Theorem~\ref{outputfeedback} in the attack detection and correction of zero-dynamics attacks \cite{teixeira2015secure}. Consider first the original system in Fig.~\ref{figure3_2}. For zero-dynamics attacks, the typical attack design is to let $Z \left( s \right) = 0$ and 
\begin{flalign} \label{zeroattack}
W \left( s \right) = \frac{w_{0}}{s-\zeta},
\end{flalign}
where $\zeta$ is a zero of $P \left( s \right)$. It is known that if $w_{0}$ is chosen correspondingly, then the attack cannot be detected, as a consequence of the blocking property of zeros.

Consider next the system with two-way coding in Fig.~\ref{injection} where the equivalent plant from the perspective of the attacker is given by $\overline{P} \left( s \right)$. If the zero-dynamics attacks are still designed in terms of the zeros of $P \left( s \right)$, then they will easily be detected as long as $c \neq 0$, since the zeros of $\overline{P} \left( s \right)$ are all different from those of $P \left( s \right)$.  

On the other hand, if the attacker somehow knows $\overline{P} \left( s \right)$ (e.g., by carrying out system identification based on $\overline{q} \left(t\right)$ and $ \overline{v} \left(t\right)$, or by knowing $a,b,c,d$ as well as $P \left( s \right)$) and designs the zero-dynamics attacks accordingly, then the attacks cannot be detected. In this case, note that if the plant $P$ is stabilizable by static output feedback, then all the zeros of $\overline{P} \left( s \right)$ can be made minimum-phase. As a result, only stable zero-dynamics attacks are possible, meaning that the attack signal and hence the attack response will be zero in steady state; in such a case, we say that the attack effect can be corrected.

We summarize the above discussions in the following corollary.

\begin{corollary} \label{zerodynamicsattack}
	Consider the system with two-way coding in Fig.~\ref{injection} under zero-dynamics attack given by \eqref{zeroattack}.
	\begin{itemize}
		\item If the zero-dynamics attack is designed according to $P \left( s \right)$, then it can always be detected with $c \neq 0$.
		\item If the zero-dynamics attack is designed with respect to $\overline{P} \left( s \right)$, then, supposing that the plant $P$ is stabilizable by static output feedback, all the zeros of $\overline{P} \left( s \right)$ can be made minimum-phase, in which case the attack effect will be corrected in steady state. 
	\end{itemize}
\end{corollary}

Note also that for zero-dynamics attacks, the attacker may instead choose to let $W \left( s \right) = 0$ and 
\begin{flalign} \label{zeroattack2}
Z \left( s \right) = \frac{z_{0}}{s-\lambda},
\end{flalign}
where $\lambda$ is a pole of $P \left( s \right)$ (and hence a zero of the closed-loop system from $z \left(t\right)$ to plant output  $\overline{y} \left(t\right)$). If $z_{0}$ is chosen correspondingly, then the attack cannot be detected. Similarly, in the system with two-way coding in Fig.~\ref{injection}, if the zero-dynamics attacks are still designed in terms of the poles of $P \left( s \right)$, they will easily be detected as long as $b \neq 0$, since the poles of $\overline{P} \left( s \right)$ are all different from those of $P \left( s \right)$. On the other hand, if the attacker knows $\overline{P} \left( s \right)$ and designs the zero-dynamics attacks accordingly, then the attacks cannot be detected. In this case, note that if the plant $P$ is stabilizable by static output feedback, then all the poles of $\overline{P} \left( s \right)$ can be made stable. As a consequence, only stable zero-dynamics attacks are possible, meaning that the attack effect will be zero in steady state; in such a situation, the attack effect is said to be corrected. 

Similarly, we summarize the previous discussions in the corollary below.

\begin{corollary} \label{zerodynamicsattack2}
	Consider the system with two-way coding in Fig.~\ref{injection} under zero-dynamics attack given by \eqref{zeroattack2}.
	\begin{itemize}
		\item If the zero-dynamics attack is designed according to $P \left( s \right)$, then it can always be detected with $b \neq 0$.
		\item If the zero-dynamics attack is designed with respect to $\overline{P} \left( s \right)$, then, supposing that the plant $P$ is stabilizable by static output feedback, all the poles of $\overline{P} \left( s \right)$ can be made stable, in which case the attack effect will be corrected in steady state. 
	\end{itemize}
\end{corollary}


When the zero-dynamics attacks \eqref{zeroattack} and \eqref{zeroattack2} happen simultaneously, it is clear that Corollary~\ref{zerodynamicsattack} and Corollary~\ref{zerodynamicsattack2} apply respectively to the two attacks.


It might also be interesting to examine what changes two-way coding can bring to the detection and correction of other classes of injection attacks; see, e.g., \cite{pasqualetti2015control}. We will, however, leave those investigations to future research.

\section{Conclusions}

We have introduced the method of two-way coding into feedback control systems under injection attacks. We have shown that the presence of two-way coding can distort the perspective of the attacker on the control system; this distorted view on the attacker side was demonstrated to facilitate detecting the attacks, or restricting what the attacker can do, or even correcting the attack effect in steady state. 
Future research directions include the analysis of MIMO systems, discrete-time systems, as well as other classes of attacks in the presence of two-way coding.  


\appendix
%
%
\begin{acks}
  The work is supported by the \grantsponsor{kth1}{Knut and Alice Wallenberg Foundation}{}, the \grantsponsor{kth1}{Swedish Strategic Research Foundation}{}, the \grantsponsor{kth1}{Swedish Research Council}{}, \grantsponsor{kth1}{the Swedish Civil Contingencies Agency (CERCES project)}{}, the \grantsponsor{tit1}{JSPS}{} under Grant-in-Aid for Scientific Research Grant No.:~\grantnum{tit1}{15H04020}, and the \grantsponsor{tit2}{JST CREST}{} under Grant No.:~\grantnum{tit2}{JPMJCR15K3}.


\end{acks}

\bibliographystyle{ACM-Reference-Format}
\bibliography{sample-bibliography}

\end{document}